\documentclass[journal,12pt,onecolumn,draftclsnofoot,]{IEEEtran}

\usepackage[pdftex]{graphicx}
\usepackage{epstopdf}
\usepackage[cmex10]{amsmath}
\usepackage{amssymb}
\usepackage{enumerate}
\usepackage{algorithm}
\usepackage{algpseudocode}
\usepackage{empheq}
\usepackage{eqparbox}
\usepackage{amsfonts}
\usepackage{amsthm}
\usepackage{url}

\usepackage{acro}

\newtheorem{thm}{Theorem}

\newtheorem{examp}{Example}

\begin{document}
	\vspace{-0.8cm}
\title{On Multi-Server Coded Caching\\ in the Low Memory Regime}

\author{
	\IEEEauthorblockN{Seyed Pooya Shariatpanahi$^\ast$, Babak Hossein Khalaj$^\dagger$} \\ \vspace{+5mm}
\IEEEauthorblockA{
    $\ast$ School of Computer Science, \\Institute for Research in Fundamental Sciences (IPM), Tehran, Iran \\
    $\dagger$ Department of Electrical Engineering,\\
    Sharif University of Technology, Tehran, Iran.
    \\
    \texttt{pooya@ipm.ir, khalaj@sharif.edu}}

\thanks{
	This research was in part supported by a grant from IPM}
}

\maketitle
\begin{abstract}
	In this paper we determine the delivery time for a multi-server coded caching problem when the cache size of each user is small.We propose an achievable scheme based on coded cache content placement, and employ zero-forcing techniques at the content delivery phase. Surprisingly, in contrast to previous multi-server results which were proved to be order-optimal within a multiplicative factor of $2$, for the low memory regime we prove that our achievable scheme is optimal. Moreover, we compare the performance of our scheme with the uncoded solution, and show our proposal's improvement over the uncoded scheme. Our results also apply to Degrees-of-Freedom (DoF) analysis of Multiple-Input Single-Output Broadcast Channels (MISO-BC) with cache-enabled users, where the multiple-antenna transmitter replaces the role of multiple servers. This shows that interference management in the low memory regime needs different caching techniques compared with medium/high memory regimes discussed in previous works.
\end{abstract}

\section{Introduction}
Caching content during network off-peak hours to relieve congestion at network high-peak hours is a well-investigated technique in the literature of content delivery networks both in wired networks (\cite{Kangasharju2002}, \cite{Borst2010}) and in wireless settings (\cite{Gitzenis2012}, \cite{Shariatpanahi2015}). Coded caching \cite{MaddahAli2014}, which has been proposed in the context of information theoretic analysis of caching networks, can be considered as a paradigm shift in this direction by providing multicasting gains (proportional to the total storage available in the network) to users with distinct demands. This approach is shown to provide substantial gains in different scenarios such as hierarchical networks \cite{Karamchandi2016}, on-line coded caching \cite{Pedarsani2016}, and D2D networks \cite{Ji2016}. 

An important line of research in the framework of coded caching is to investigate how one can use multiple transmitters to boost the coded caching scheme performance. This problem has been considered in the context of wired networks under the name of \emph{multi-server coded caching}  \cite{Shariatpanahi2016}, \cite{Lampiris2018}, \cite{Mital2018}, \cite{Cheng2018}, and in the context of wireless networks under the names of \emph{MISO-BC networks} \cite{Zhang2017}, \cite{Shariatpanahi2017}, \emph{cache-enabled interference networks} \cite{Naderalizadeh2017}, \cite{Cao}, \cite{Tahmasbi2017}, and \emph{multi-antenna coded caching} \cite{Shariatpanahi2017ISIT}. The interesting result in \cite{Shariatpanahi2016} (and follow-up works) shows that with multiple transmitters the multiplexing gain offered by the transmitters and the multicasting gain of coded caching are \emph{additive}, which is applicable to all the above wired and wireless multi-transmitter setups. This encouraging result suggests that using multiple transmitters along with coded caching techniques will guarantee high data rates needed for future wireless content delivery applications.

In this paper we consider a multi-server coded caching setup where, in contrast to previous works, the cache size of each user is much smaller than the size of a single file. In this regime the cache content placement scheme used is storing a linear combination of sub-files in users' caches. The delivery phase would benefit from zero-forcing techniques. Interestingly, we show that this strategy is optimal, by presenting a matching converse proof for the delivery time. It should be noted that our paper can be considered as a generalization of the work \cite{Chen2014}, which considers a low memory regime in a single-server setup, to the multiple-server setup. The structure of the paper is as follows. In Section \ref{Sec_SystemModel} we present the problem setup. In section \ref{Sec_LowMemory} we consider the problem in the low memory regime which contains two subsections, each investigating a different regime for the number of antennas. Finally, section \ref{Sec_Conclusions} concludes the paper.

\section{System Model}\label{Sec_SystemModel}

We consider $L$ transmitters sending data to $K$ cache-enabled users via a \emph{Linear Network}. All the transmitters are assumed to have access to a library of $N$ files $\mathcal{W}=\{W_1,\dots,W_N\}$, each of $F$ bits.  This is a general model which covers a wired network setup where a single server is connected to an intermediate network with $L$ unit-capacity links (or equivalently $L$ servers each with a  unit-capacity links), and each user is connected to the network with one unit-capacity link. Moreover, internal nodes do random linear network coding resulting in a linear network (see \cite{Shariatpanahi2016}). Alternatively, this model covers a wireless Multi-Input Single-Output Broadcast (MISO-BC) setup, where a multi-antenna base station with $L$ antennas delivers content to $K$ single-antenna users (see \cite{Zhang2017}). 

Based on the \emph{Linear Network} assumption mentioned above, if the transmit vector in time slot $i$ is $\mathbf{x}(i)$, the received signal at user $k$ will be:
\begin{equation}
y_k(i)=\mathbf{h}_k^{H} \mathbf{x}(i), \quad k=1,\dots,K
\end{equation}
where $\mathbf{h}_k$ is the channel vector from the $L$ transmitters to the user $k$. In this paper we omit the time slot index whenever it is clear from the context.

The network operates in two phase, namely, the \emph{Cache Content Placement} and the \emph{Content Delivery} phases. In the first phase, which is assumed to happen during network low-peak hours, the users' caches are filled with data from the library. More specifically, we denote the cache content of user $k$ as $Z_k$ to be a function of the library $\mathcal {W}$, which should have an entropy less than $MF$ bits. It should be noted that this phase is operated without knowing users content requests in the content delivery phase. In the second phase, which is assumed to occur during network high-peak hours, each user requests a content from the library, denoted collectively by the indexes vector $\mathbf{d}=\{d_1,\dots,d_K\}$, where $d_k \in [N]$ denotes the request index of user $k$. In order to assume the worst case request vector $\mathbf{d}$, and remove any non-coded multicasting opportunities, we assume that all the users request distinct files from the library (i.e., $d_i \neq d_j$ if $i\neq j$). According to these requests, the transmitters collaboratively send a space-time block code $\mathbf{X}(\mathbf{d})$ of size $L \times T$, such that each user can decode its requested file with the help of its received signal at the second phase, along with its cache contents acquired in the first phase. We define the \emph{Delivery Time} $T$ to be the number of network/channel uses needed to transmit $\mathbf{X}$ as the performance metric for the caching schemes.

\section{The Low Memory Regime}\label{Sec_LowMemory}

In this section we consider the performance of network in the regime of low memory. More specifically, we assume $K=N$ and $M=1/N$. Thus, each user can cache only a fraction of each file. In the first subsection we assume $L=N-1$, and propose a scheme which achieves the optimal performance. Then in the next subsection we investigate the case of $L<N-1$.

\subsection{Problem Parameters: $K=N$ $M=1/N$, $L=N-1$}\label{SubSec_LowMemory1}

Let us begin explaining the main idea via an example\footnote{It should be noted that the example of $K=N=3$, $M=1/3$, and $L=2$ is investigated in \cite{Shariatpanahi2016}.}:

\begin{examp}
	In this example we assume $L=3$ transmitters, $K=N=4$ receivers and files, and $M=\frac{1}{4}$. Let us denote the files as $A$, $B$, $C$, and $D$. In the cache content placement the users' caches are filled as follows:
	\begin{eqnarray}
	Z_1&=&\{A_1+B_1+C_1+D_1\} \\ \nonumber
	Z_2&=&\{A_2+B_2+C_2+D_2\} \\ \nonumber
	Z_3&=&\{A_3+B_3+C_3+D_3\} \\ \nonumber
	Z_4&=&\{A_4+B_4+C_4+D_4\} 
	\end{eqnarray}
	Suppose in the second phase the first, second, third, and fourth users request files $A$, $B$, $C$, and $D$ respectively. The signal transmitted by the transmitters will be:
	\begin{eqnarray}
		\mathbf{X}_1&=&B_1 \frac{\mathbf{w}_2^{\{2,3,4\}}}{\mathbf{h}_1^{H}\mathbf{w}_2^{\{2,3,4\}}}+C_1 \frac{\mathbf{w}_3^{\{2,3,4\}}}{\mathbf{h}_1^{H}\mathbf{w}_3^{\{2,3,4\}}}+D_1 \frac{\mathbf{w}_4^{\{2,3,4\}}}{\mathbf{h}_1^{H}\mathbf{w}_4^{\{2,3,4\}}} \\ \nonumber
		\mathbf{X}_2&=&A_2 \frac{\mathbf{w}_1^{\{1,3,4\}}}{\mathbf{h}_2^{H}\mathbf{w}_1^{\{1,3,4\}}}+C_2 \frac{\mathbf{w}_3^{\{1,3,4\}}}{\mathbf{h}_2^{H}\mathbf{w}_3^{\{1,3,4\}}}+D_2 \frac{\mathbf{w}_4^{\{1,3,4\}}}{\mathbf{h}_2^{H}\mathbf{w}_4^{\{1,3,4\}}} \\ \nonumber
		\mathbf{X}_3&=&A_3 \frac{\mathbf{w}_1^{\{1,2,4\}}}{\mathbf{h}_3^{H}\mathbf{w}_1^{\{1,2,4\}}}+B_3 \frac{\mathbf{w}_2^{\{1,2,4\}}}{\mathbf{h}_3^{H}\mathbf{w}_2^{\{1,2,4\}}}+D_3 \frac{\mathbf{w}_4^{\{1,2,4\}}}{\mathbf{h}_3^{H}\mathbf{w}_4^{\{1,2,4\}}} \\ \nonumber
		\mathbf{X}_4&=&A_4 \frac{\mathbf{w}_1^{\{1,2,3\}}}{\mathbf{h}_4^{H}\mathbf{w}_1^{\{1,2,3\}}}+B_4 \frac{\mathbf{w}_2^{\{1,2,3\}}}{\mathbf{h}_4^{H}\mathbf{w}_2^{\{1,2,3\}}}+C_4 \frac{\mathbf{w}_3^{\{1,2,3\}}}{\mathbf{h}_4^{H}\mathbf{w}_3^{\{1,2,3\}}} 
	\end{eqnarray}
	The unit-size vectors $\mathbf{w}^{S}_i$ are chosen such that $\mathbf{h}_j^H \mathbf{w}^{S}_i =0$ for all $j \in S \backslash \{i\}$. 
	Let us focus on the received signals by all the users after transmission of $\mathbf{X}_1$:
	\begin{eqnarray}
	&&\mathbf{h}_2^H\mathbf{X}_1=B_1 \frac{\mathbf{h}_2^H\mathbf{w}_2^{\{2,3,4\}}}{\mathbf{h}_1^{H}\mathbf{w}_2^{\{2,3,4\}}} \quad
	\mathbf{h}_3^H\mathbf{X}_1=C_1 \frac{\mathbf{h}_3^H\mathbf{w}_3^{\{2,3,4\}}}{\mathbf{h}_1^{H}\mathbf{w}_3^{\{2,3,4\}}} \quad
	\mathbf{h}_4^H\mathbf{X}_1=D_1 \frac{\mathbf{h}_4^H\mathbf{w}_4^{\{2,3,4\}}}{\mathbf{h}_1^{H}\mathbf{w}_4^{\{2,3,4\}}} \\ \nonumber
	&&\mathbf{h}_1^H\mathbf{X}_1=B_1 \frac{\mathbf{h}_1^H\mathbf{w}_2^{\{2,3,4\}}}{\mathbf{h}_1^{H}\mathbf{w}_2^{\{2,3,4\}}}+C_1 \frac{\mathbf{h}_1^H\mathbf{w}_3^{\{2,3,4\}}}{\mathbf{h}_1^{H}\mathbf{w}_3^{\{2,3,4\}}}+D_1 \frac{\mathbf{h}_1^H\mathbf{w}_4^{\{2,3,4\}}}{\mathbf{h}_1^{H}\mathbf{w}_4^{\{2,3,4\}}} =B_1 + C_1 + D_1
	\end{eqnarray}
	By transmitting $\mathbf{X}_2$ we will have:
	\begin{eqnarray}
	&&\mathbf{h}_1^H\mathbf{X}_2=A_2 \frac{\mathbf{h}_1^H\mathbf{w}_1^{\{1,3,4\}}}{\mathbf{h}_2^{H}\mathbf{w}_1^{\{1,3,4\}}} \quad
	\mathbf{h}_3^H\mathbf{X}_2=C_2 \frac{\mathbf{h}_3^H\mathbf{w}_3^{\{1,3,4\}}}{\mathbf{h}_2^{H}\mathbf{w}_3^{\{1,3,4\}}} \quad
	\mathbf{h}_4^H\mathbf{X}_2=D_2 \frac{\mathbf{h}_4^H\mathbf{w}_4^{\{1,3,4\}}}{\mathbf{h}_2^{H}\mathbf{w}_4^{\{1,3,4\}}} \\ \nonumber
	&&\mathbf{h}_2^H\mathbf{X}_2=A_2 \frac{\mathbf{h}_2^H\mathbf{w}_1^{\{1,3,4\}}}{\mathbf{h}_2^{H}\mathbf{w}_1^{\{1,3,4\}}}+C_2 \frac{\mathbf{h}_2^H\mathbf{w}_3^{\{1,3,4\}}}{\mathbf{h}_2^{H}\mathbf{w}_3^{\{1,3,4\}}}+D_2 \frac{\mathbf{h}_2^H\mathbf{w}_4^{\{1,3,4\}}}{\mathbf{h}_2^{H}\mathbf{w}_4^{\{1,3,4\}}} = A_2 + C_2 + D_2
	\end{eqnarray}
	By transmitting $\mathbf{X}_3$ we will have:
	\begin{eqnarray}
	&&\mathbf{h}_1^H\mathbf{X}_3=A_3 \frac{\mathbf{h}_1^H\mathbf{w}_1^{\{1,2,4\}}}{\mathbf{h}_3^{H}\mathbf{w}_1^{\{1,2,4\}}} \quad
	\mathbf{h}_2^H\mathbf{X}_3=B_3 \frac{\mathbf{h}_2^H\mathbf{w}_2^{\{1,2,4\}}}{\mathbf{h}_3^{H}\mathbf{w}_2^{\{1,2,4\}}} \quad
	\mathbf{h}_4^H\mathbf{X}_3=D_3 \frac{\mathbf{h}_4^H\mathbf{w}_4^{\{1,2,4\}}}{\mathbf{h}_3^{H}\mathbf{w}_4^{\{1,2,4\}}} \\ \nonumber
	&&\mathbf{h}_3^H\mathbf{X}_3=A_3 \frac{\mathbf{h}_3^H\mathbf{w}_1^{\{1,2,4\}}}{\mathbf{h}_3^{H}\mathbf{w}_1^{\{1,2,4\}}}+B_3 \frac{\mathbf{h}_3^H\mathbf{w}_2^{\{1,2,4\}}}{\mathbf{h}_3^{H}\mathbf{w}_2^{\{1,2,4\}}}+D_3 \frac{\mathbf{h}_3^H\mathbf{w}_4^{\{1,2,4\}}}{\mathbf{h}_3^{H}\mathbf{w}_4^{\{1,2,4\}}} = A_3 + B_3 + D_3
	\end{eqnarray}
	And finally by transmitting $\mathbf{X}_4$ we have:
		\begin{eqnarray}
	&&\mathbf{h}_1^H\mathbf{X}_4=A_4 \frac{	\mathbf{h}_1^H\mathbf{w}_1^{\{1,2,3\}}}{\mathbf{h}_4^{H}\mathbf{w}_1^{\{1,2,3\}}} ,\quad
	\mathbf{h}_2^H\mathbf{X}_4=B_4 \frac{\mathbf{h}_2^H\mathbf{w}_2^{\{1,2,3\}}}{\mathbf{h}_4^{H}\mathbf{w}_2^{\{1,2,3\}}} ,\quad
	\mathbf{h}_3^H\mathbf{X}_4=C_4 \frac{\mathbf{h}_3^H\mathbf{w}_3^{\{1,2,3\}}}{\mathbf{h}_4^{H}\mathbf{w}_3^{\{1,2,3\}}} \\ \nonumber
	&&\mathbf{h}_4^H\mathbf{X}_4=A_4 \frac{\mathbf{h}_4^H\mathbf{w}_1^{\{1,2,3\}}}{\mathbf{h}_4^{H}\mathbf{w}_1^{\{1,2,3\}}}+B_4 \frac{\mathbf{h}_4^H\mathbf{w}_2^{\{1,2,3\}}}{\mathbf{h}_4^{H}\mathbf{w}_2^{\{1,2,3\}}}+C_4 \frac{\mathbf{h}_4^H\mathbf{w}_3^{\{1,2,3\}}}{\mathbf{h}_4^{H}\mathbf{w}_3^{\{1,2,3\}}} =A_4+B_4+C_4
	\end{eqnarray}
By collecting all the decoded sub-files we arrive at the below table which shows the decoded data by each user following each transmission. We call this table as the \emph{Delivery Table} for this problem.
\begin{center}
	\begin{tabular}{ |c| c | c | c | c | c | c |}
		\hline
		Row & Signal & User 1 & User 2 & User 3 & User 4 & Time Slot\\ \hline
		1 & $\mathbf{X}_1$ & $B_1+C_1+D_1$ & $B_1$ & $C_1$ & $D_1$ & $\frac{1}{4}$ \\ \hline
		2 & $\mathbf{X}_2$ &  $A_2$ & $A_2+C_2+D_2$ & $C_2$ & $D_2$ & $\frac{1}{4}$ \\ \hline
		3 & $\mathbf{X}_3$ & $A_3$ & $B_3$ & $A_3+B_3+D_3$ & $D_3$ & $\frac{1}{4}$ \\
		\hline
		4 & $\mathbf{X}_4$ & $A_4$ & $B_4$ & $C_4$& $A_4+B_4+C_4$ & $\frac{1}{4}$ \\
		\hline
	\end{tabular}
\end{center}
Then, it is clear that each user can decode its requested file with the help of its cache contents. Since each row in the delivery table takes $1/N=1/4$ time slots, sending the transmit blocks $\mathbf{X}_1$, $\mathbf{X}_2$, $\mathbf{X}_3$, and $\mathbf{X}_4$ will result in the Delivery Time of 
\begin{equation}
T=4\times \frac{1}{4} =1
\end{equation}
Now, following Lemma 1 in \cite{Shariatpanahi2016} we have the following lower bound on the delivery time
\begin{eqnarray}
T &\geq& \max_{s \in \{1,\dots,K\}} \frac{1}{\min(s,L)} \left(s-\frac{s}{\lfloor N/s \rfloor }M\right) \\ \nonumber
&\geq&\frac{1}{L}\left(K-\frac{K}{\lfloor N/K \rfloor }M\right) \\ \nonumber
&=&1
\end{eqnarray}
which shows that the above achievable scheme is optimal.	

The above delivery delay of $T=1$ should be compared to the uncoded scheme in which every user caches $M/N$ fraction of each file. Thus, by applying the classical Zero-Forcing and forming $L$ parallel streams, the Delivery Time will be
\begin{eqnarray}
 T&=&\frac{K(1-M/N)}{L} \\ \nonumber
 &=&\frac{5}{4}
\end{eqnarray}
which shows that the optimal proposed scheme will result in $\frac{1}{4}$ time slots less delay.
\end{examp}

As we see next, the same concept of Delivery Table can be extended to other examples as well.

\begin{examp}
	In this example we assume $L=4$ transmitters, $K=N=5$ receivers and files, and $M=\frac{1}{4}$. Let us denote the files as $A$, $B$, $C$,  $D$, and $E$, and suppose the users require them respectively. In the cache content placement the users' caches are filled as follows:
	\begin{eqnarray}
	Z_i=\{A_i+B_i+C_i+D_i+E_i\} 
	\end{eqnarray}
	for $i=1,\dots,5$.	Along the same guidelines provided in Example 1 one can arrive at the following delivery table for this example. 
	\begin{center} \small
	\begin{tabular}{ | c | c | c | c | c | c |}
		\hline
		Signal & User 1 & User 2 & User 3 & User 4 & User 5 \\ \hline
		$\mathbf{X}_1$ & $B_1+C_1+D_1+E_1$ & $B_1$ & $C_1$ & $D_1$ & $E_1$  \\ \hline
		$\mathbf{X}_2$ & $A_2$ & $A_2+C_2+D_2+E_2$ & $C_2$ & $D_2$ & $E_2$  \\ \hline
		$\mathbf{X}_3$ & $A_3$ & $B_3$ & $A_3+B_3+D_3+E_3$ & $D_3$ & $E_3$  \\ \hline
		$\mathbf{X}_4$ & $A_4$ & $B_4$ & $C_4$ & $A_4+B_4+C_4+E_4$ & $E_4$  \\ \hline
		$\mathbf{X}_5$ & $A_5$ & $B_5$ & $C_5$ & $D_5$ & $A_5+B_5+C_5+D_5$  \\ \hline
		
	\end{tabular}
\end{center}
Then one can easily arrive at the delivery time of $T=5 \times \frac{1}{5}=1$ which is optimal. The Delivery Time for the uncoded scheme will be
\begin{eqnarray}
T&=&\frac{K(1-M/N)}{L} \\ \nonumber
&=&\frac{6}{5}
\end{eqnarray}
which shows that the optimal proposed scheme will result in $\frac{1}{5}$ time slots less delay.
\end{examp}

The following theorem generalizes the above examples.
\begin{thm}
	Suppose $K=N$, $L=N-1$, and $M=\frac{1}{N}$. Then, the optimal delivery time is $T=1$.
\end{thm}

\begin{proof}
	 Let us present our achievable scheme in Algorithm I.

	\begin{algorithm}[t]
		\caption{Multi-Server Coded Caching for Small Cache Size\label{Alg_Main}}
		\begin{algorithmic}[1]
			\Procedure{CACHE-PLACEMENT}{$W_1,\dots,W_N$}
			
			\ForAll{$n=1,\dots,N$}
			\State	$W_n=\{W_n^i\}$ for $i=1,\dots,N$
			\EndFor
			
			\ForAll{$k=1,\dots,K$}
			\State $Z_k=\sum_{n=1}^{N} W_{n}^k$
			\EndFor
			
			\EndProcedure
			\Procedure{CONTENT-DELIVERY}{$W_1,\dots,W_N$, $d_1,\dots,d_K$, $\mathbf{H} = [\mathbf{h}_1, \ldots, \mathbf{h}_K]$}
			\ForAll{$i=1,\dots,K$}
			\State $\mathbf{X}_i \leftarrow \sum_{k=1, k\neq i}^K W_{d_k}^i \frac{\mathbf{w}_k^{[K]\backslash\{i\}}}{\mathbf{h}_i^H\mathbf{w}_k^{[K]\backslash\{i\}}} $ where $\mathbf{h}_j^H\mathbf{w}^S_k=0$ for all $ j\in S \backslash \{k\}$ 
			\State Transmit $\mathbf{X}_i$
			
			\EndFor
			
			\EndProcedure
			
		\end{algorithmic}
	\end{algorithm}
	Next we show that Algorithm 1 delivers all the desired requests to the users correctly. Let us focus on an arbitrary user $j$ which has requested the file $W_{d_j}$. Upon transmission of $\mathbf{X}_i$ for $i\neq j$ this user receives
	\begin{eqnarray}
	\mathbf{h}_j^H \mathbf{X}_i&=&\sum_{k=1, k\neq i}^K W_{d_k}^i \frac{\mathbf{h}_j^H\mathbf{w}_k^{[K]\backslash\{i\}}}{\mathbf{h}_i^H\mathbf{w}_k^{[K]\backslash\{i\}}} \\ \nonumber
	&=& W_{d_j}^i \frac{\mathbf{h}_j^H\mathbf{w}_j^{[K]\backslash\{i\}}}{\mathbf{h}_i^H\mathbf{w}_j^{[K]\backslash\{i\}}}
	\end{eqnarray}
since $\mathbf{h}_j^H\mathbf{w}_j^{[K]\backslash\{i\}} \neq 0$ and $\mathbf{h}_i^H\mathbf{w}_j^{[K]\backslash\{i\}}\neq 0 $ with high probability, user $j$ can decode $W_{d_j}^i$ for all $i \in [N] \backslash \{j\}$. So for decoding the whole file it remains for this user to decode $W_{d_j}^j$.

Now let us focus on what this user receives after transmission of $\mathbf{X}_j$:
\begin{eqnarray}
\mathbf{h}_j^H \mathbf{X}_j&=&\sum_{k=1, k\neq j}^K W_{d_k}^j \frac{\mathbf{h}_j^H\mathbf{w}_k^{[K]\backslash\{j\}}}{\mathbf{h}_j^H\mathbf{w}_k^{[K]\backslash\{j\}}} \\ \nonumber
&=& \sum_{k=1, k\neq j}^K W_{d_k}^j 
\end{eqnarray}
by subtracting this from $Z_j$ we will have:
\begin{equation}
\sum_{n=1}^{N} W_{n}^j -\sum_{k=1, k\neq j}^K W_{d_k}^j =W_{d_j}^j
\end{equation}
which is the missing part. Thus user $j$ can decode $W_{d_j}$, and similarly, all the users can decode their requests. 

The Delivery Time of this achievable can be calculated as the number of transmit blocks $\mathbf{X}_i$, which is $N$, times the delivery time of each, which is $1/N$, resulting in $T=1$. 

Finally, following from  the converse Lemma 1 in \cite{Shariatpanahi2016} we have 
\begin{eqnarray}
T &\geq& \max_{s \in \{1,\dots,K\}} \frac{1}{\min(s,L)} \left(s-\frac{s}{\lfloor N/s \rfloor }M\right) \\ \nonumber
&\geq&\frac{1}{L}\left(K-\frac{K}{\lfloor N/K \rfloor }M\right) \\ \nonumber
&=&\frac{N-1}{L} \\ \nonumber
&=&1
\end{eqnarray}
which concludes the proof.

\end{proof}

In comparison with the uncoded scheme which arrives at the delivery time of
\begin{eqnarray}
T&=&\frac{K(1-M/N)}{L} \\ \nonumber
&=&1+\frac{1}{N}
\end{eqnarray}
we see $1/N$ time slots improvement in the delivery time.

\subsection{Problem Parameters: $K=N$ $M=1/N$, $L<N-1$}\label{SubSec_LowMemory2}

In the last subsection we observed that as long as we have $L=N-1$ antennas, each row of the delivery table can be delivered in one shot of length  $1/N$ time slots. However, when we have less antennas, delivery of each row is different. In each row of the delivery table the goal is to deliver $N-1$ individual messages to $N-1$ of the users and the sum of these messages to the remaining user. For example, in the first row of Example 1's delivery table there are three individual messages for the second, third, and the fourth users, and the sum of these messages should be delivered to the first user. As we have shown in the previous subsection, this is feasible if we have $L=N-1$ transmitters. Next, we explain how the achievable scheme changes if we have less antennas.

\begin{examp}
	The setup of this example is the same as Example 1 except that now we have $L=2$ antennas. Suppose the goal is to deliver $M_1$ to the user 1, $M_2$ to the user 2, $M_3$ to the user 3, and $M_1+M_2+M_3$ to the user 4. All $M_i$'s have the length of $1/N=1/4$, thus, with three transmitters we could fulfill this task in one shot of length $1/4$. However in order to do this with $L=2$ antennas first we need to further split each sub-file into two equal mini-files, i.e., $M_i=\{M_i^1,M_i^2\}, i=1,2,3$. Then, we send the following signals
	\begin{eqnarray}
	&& M_1^1 \frac{\mathbf{h}_2^{\perp}}{\mathbf{h}_4^H\mathbf{h}_2^{\perp}} + (M_2^1 + M_2^2) \frac{\mathbf{h}_1^{\perp}}{\mathbf{h}_4^H\mathbf{h}_1^{\perp}}\\ \nonumber
	&& M_2^2 \frac{\mathbf{h}_3^{\perp}}{\mathbf{h}_4^H\mathbf{h}_3^{\perp}} - M_3^1 \frac{\mathbf{h}_2^{\perp}}{\mathbf{h}_4^H\mathbf{h}_2^{\perp}}\\ \nonumber
	&& M_1^2 \frac{\mathbf{h}_3^{\perp}}{\mathbf{h}_4^H\mathbf{h}_3^{\perp}} + (M_3^1+M_3^2) \frac{\mathbf{h}_1^{\perp}}{\mathbf{h}_4^H\mathbf{h}_1^{\perp}}
	\end{eqnarray}
	It can be easily checked that the data different users receive are as summarized in the table below
		\begin{center}
		\begin{tabular}{  | c | c | c | c | c |}
			\hline
			 User 1 & User 2 & User 3 & User 4 & Time Slot\\ \hline
			$M_1^1$ & $M_2^1+M_2^2$ & - & $M_1^1+M_2^1+M_2^2$ & $\frac{1}{8}$ \\ \hline
			 - & $M_2^2$ & $-M_3^1$ & $M_2^2-M_3^1$ & $\frac{1}{8}$ \\ \hline
		 $M_1^2$ & - & $M_3^1+M_3^2$ & $M_1^2+M_3^1+M_3^2$ & $\frac{1}{8}$ \\
			\hline
			
		\end{tabular}
	\end{center}

Now it is clear that user 1 can decode $M_1=\{M_1^1,M_1^2\}$, user 2 can decode $M_2=\{M_2^1,M_2^2\}$, and user 3 can decode $M_3=\{M_3^1,M_3^2\}$. Also by subtracting the second row from the first row, user 4 can decode $M_1^1+M_2^1+M_3^1$. User 4 can also add the second and third rows to arrive at $M_1^2+M_2^2+M_3^2$. Thus, user 4 can arrive at $M_1+M_2+M_3=\{M_1^1+M_2^1+M_3^1,M_1^2+M_2^2+M_3^2\}$. The total time for finishing this task is $3 \times 1/8=3/8$ time slots achieved by $L=2$ in contrast to $1/4$ time slots achieved in Example 1 by $L=3$ antennas, i.e., a multiplicative factor of $3/2$ more time slots needed due to less transmitters available. Since all the rows in the delivery table of Example 1 can be treated similarly, the total time needed is now $T=\frac{3}{2}$.

On the other hand from the converse argument in Theorem 1 we know 
\begin{equation}
T \geq \frac{N-1}{L} = \frac{3}{2}
\end{equation}
which shows that the proposed scheme is optimal. The Delivery Time for the uncoded scheme will be
\begin{eqnarray}
T&=&\frac{K(1-M/N)}{L} \\ \nonumber
&=&\frac{15}{8}
\end{eqnarray}
which shows that the optimal proposed scheme will result in $\frac{3}{8}$ time slots less delay.

\end{examp}

\begin{examp}
	Here we revisit Example 2 with $L=3$ antennas. Each row in Example 1 consisted of delivering an independent message to one of the four users, and the sum of these messages to the remaining user. Assume we want to deliver $M_1$, $M_2$, $M_3$, and $M_4$ to the users 1, 2, 3, 4. In addition the message $M_1+M_2+M_3+M_4$ should be delivered to the user 5. Since we have $L=3$ transmitters, we can send three independent messages in parallel. Here we split each message into three equal parts, i.e., $M_i=\{M_i^1,M_i^2,M_i^3\}$. Then, the below table shows the coding strategy in this case 
	\begin{center} \small
		\begin{tabular}{  | c | c | c | c | c |}
			\hline
			 User 1 & User 2 & User 3 & User 4 & User 5 \\ \hline
			 $M_1^1-M_1^2+M_1^3$ & $M_2^1-M_2^2$ & $M_3^1$ & - & $M_1^1-M_1^2+M_1^3+M_2^1-M_2^2+M_3^1$  \\ \hline
			 $M_1^2-M_1^3$ & $M_2^2$ & - & $M_4^1$ & $M_1^2-M_1^3+M_2^2+M_4^1$  \\ \hline
			 $M_1^3$ & - & $M_3^2$ & $M_4^2-M_4^1$ & $M_1^3+M_3^2+M_4^2-M_4^1$  \\ \hline
			 - & $M_2^3$ & $M_3^3-M_3^2$ & $M_4^3-M_4^2+M_4^1$ & $M_2^3+M_3^3-M_3^2+M_4^3-M_4^2+M_4^1$  \\ \hline
		\end{tabular}
	\end{center}

It is clear that user 1 can decode $M_1=\{M_1^1,M_1^2,M_1^3\}$, user 2 can decode $M_2=\{M_2^1,M_2^2,M_2^3\}$, user 3 can decode $M_3=\{M_3^1,M_3^2,M_3^3\}$, and user 4 can decode $M_4=\{M_4^1,M_4^2,M_4^3\}$. Also, by adding the first and second rows, user 5 can decode $M_1^1+M_2^1+M_3^1+M_4^1$, by adding the second and third rows user 5 can decode $M_1^2+M_2^2+M_3^2+M_4^2$, and by adding the third and fourth rows user 5 can decode $M_1^3+M_2^3+M_3^3+M_4^3$. Thus, user 5 can collectively arrive at $M_1+M_2+M_3+M_4$ which was desired. The whole task of sending this single row is fulfilled in $4 \times \frac{1}{3} \times \frac{1}{5}$, which is a $\frac{4}{3}$ multiplicative factor worse than the $L=4$ case. Thus, the total delivery time will be $T=\frac{4}{3}$, which matches the converse of $T \geq \frac{N-1}{L}=\frac{4}{3}$.

Also, for the uncoded scheme we will have
\begin{eqnarray}
T&=&\frac{K(1-M/N)}{L} \\ \nonumber
&=&\frac{8}{5}
\end{eqnarray}
\end{examp}

The following theorem characterizes the optimal delivery time for all $L<N-1$ if $N-1$ is a multiple of $L$. 
\begin{thm}
		Suppose $K=N$, $M=\frac{1}{N}$, and $L$ divides $N-1$. Then, the optimal delivery time is $T=\frac{N-1}{L}$.
\end{thm}
\begin{proof}
	If we had $L=1$, then each row of the delivery table would take $(N-1) \times \frac{1}{N}$ time slots. With $L$ transmitters available, we can group the users which require independent messages in groups of size $L$ and use zero-forcing to remove intra-group interference. This will reduce the transmission of each to take $\frac{1}{L} \times (N-1) \times \frac{1}{N}$. Since we have a total of $N$ rows, the total time needed would be $T=\frac{N-1}{L}$. The converse argument is identical to that of Theorem 1 which shows the optimality of the scheme.
\end{proof}

The uncoded scheme arrives at the delivery time of
\begin{eqnarray}
T&=&\frac{K(1-M/N)}{L} \\ \nonumber
&=&\left(1+\frac{1}{N}\right) \left(\frac{N-1}{L}\right)
\end{eqnarray}
which is greater than our proposed scheme's delay.

\section{Conclusions}\label{Sec_Conclusions}
We have characterized the optimal delivery time of coded caching in multi-server networks in the low memory regime. Our achievable scheme includes caching coded content, and using zero forcing at the content delivery phase. Our converse matches the achievable scheme's performance which ensures its optimality. Also, we have compared the delivery time of our proposal with the conventional uncoded scheme, where every user caches a fraction of each file separately, and have shown our  proposal's superiority. The results can also be interpreted as DoF performance of multiple-antenna coded caching schemes, and cache-enabled interference channels where the transmitters play the role of a distributed MIMO transmitter.

\bibliographystyle{IEEEtran}

\end{document}